\documentclass[12pt,a4paper,english]{article}
\usepackage[ansinew]{inputenc}
\usepackage[T1]{fontenc}
\usepackage[english]{babel}
\usepackage{tabularx}
\usepackage{fancybox}
\usepackage{fancyhdr}
\usepackage{color}
\usepackage{authblk}
\usepackage{setspace}
\usepackage[numbers, sort]{natbib}

\bibpunct{[}{]}{,}{n}{,}{,}
\usepackage{amsmath}
\usepackage{amstext}
\usepackage{amssymb}
\usepackage{amsthm}
\usepackage{mathrsfs}
\newcommand{\be}{\begin{equation}}
\newcommand{\ee}{\end{equation}}
\newcommand{\ba}{\begin{aligned}}
\newcommand{\ea}{\end{aligned}}
\newcommand{\R}{\mathbb{R}}
\newcommand{\N}{\mathbb{N}}
\newcommand{\FF}{\mathbb{F}}
\newcommand{\F}{\mathcal{F}}
\newcommand{\FFtilde}{\widetilde{\mathbb{F}}}
\newcommand{\Ftilde}{\widetilde{\mathcal{F}}}
\newcommand{\Wtilde}{\widetilde{W}}
\newcommand{\Stilde}{\widetilde{S}}
\newcommand{\A}{\mathcal{A}}
\newcommand{\E}{\mathcal{E}}
\newcommand{\D}{\mathcal{D}}
\newcommand{\dd}{{\rm d}}
\newcommand{\ind}{\mathbf{1}}

\newcommand{\lsi}{\left[\negthinspace\left[}
\newcommand{\rsi}{\right]\negthinspace\right]}

\DeclareMathOperator{\esssup}{ess\,sup}
\DeclareMathOperator{\sign}{sign}
\newtheorem{theorem}{\bf Theorem}[section]
\newtheorem{definition}[theorem]{\bf Definition}
\newtheorem{proposition}[theorem]{\bf Proposition}
\newtheorem{lemma}[theorem]{\bf Lemma}

\theoremstyle{remark}
\newtheorem{remark}{\bf Remark}[section]
\numberwithin{equation}{section}
\usepackage{geometry}
\makeatletter
\renewcommand*\@fnsymbol[1]{\the#1}
\makeatother

\title{A note on arbitrage, approximate arbitrage \\ and the fundamental theorem of asset pricing}

\author{\vspace{0.2cm} Claudio Fontana}
\affil{\normalsize{Laboratoire Analyse et Probabilit\'e\\
Universit\'e d'\'Evry Val d'Essonne,\\
23 bd de France, 91037, \'Evry (France)} \\\vspace{0.2cm}
E-mail: \texttt{claudio.fontana@univ-evry.fr}}

\date{This version: November 27, 2013}

\begin{document}

\maketitle

\abstract{\begin{spacing}{1.1}\noindent We provide a critical analysis of the proof of the fundamental theorem of asset pricing given in the paper \emph{Arbitrage and approximate arbitrage: the fundamental theorem of asset pricing} by B. Wong and C.C. Heyde (\emph{Stochastics}, 2010) in the context of incomplete It\^o-process models. We show that their approach can only work in the known case of a complete financial market model and give an explicit counterexample.\end{spacing}}
\vspace{0.5cm}

\begin{spacing}{1}\noindent \textbf{Keywords:} arbitrage; fundamental theorem of asset pricing; It\^o-process; complete market; equivalent local martingale measure; martingale deflator.\\[-0.2cm]

\end{spacing}

\section{Introduction}	\label{S1}

One of the central results of mathematical finance is the \emph{Fundamental Theorem of Asset Pricing (FTAP)}, which, in the case of locally bounded processes, asserts the equivalence between the \emph{No Free Lunch with Vanishing Risk (NFLVR)} condition and the existence of an \emph{Equivalent Local Martingale Measure (ELMM)} with respect to which discounted asset prices are local martingales (see \cite{DS1}, Corollary 1.2). Passing from local martingales to $\sigma$-martingales, this fundamental result has been then extended to general semimartingale models in \cite{DS3}. Let us recall that NFLVR reinforces the classical \emph{No Arbitrage} condition and is equivalent to the notion of \emph{No Approximate Arbitrage} considered in \cite{LS,WH} (see e.g. \cite{F}, Lemma 6.2).

Despite its importance, the proof of the FTAP given in \cite{DS1,DS3} for continuous-time models has not been successfully simplified during the last two decades. To the best of our knowledge, the only exceptions are the classical paper \cite{LS}, where the FTAP is proved by relying on purely probabilistic arguments in the context of a complete It\^o-process model, the paper \cite{KarL}, in the context of exponential L\'evy models, and the recent paper \cite{L}, where the author succeeds in presenting a transparent proof of the FTAP for continuous semimartingales whose characteristics are absolutely continuous with respect to the Lebesgue measure.

Recently, a probabilistic and simple proof of the FTAP in the context of general incomplete financial markets based on It\^o-processes has been proposed by \cite{WH}, thus extending considerably the analysis of \cite{LS}. It is well-known that the difficult step in the proof of the FTAP consists in showing that the absence of arbitrage (in the sense of NFLVR) implies the existence of an ELMM. The proof given in \cite{WH} relies on a general characterisation of attainable claims together with the closedness under pairwise maximisation of a family of stochastic exponentials, corresponding to the candidate density processes of ELMMs (see \cite{WH}, Sections 3-4).

In the present paper, we show that the technique adopted by \cite{WH} fails to provide a new proof of the FTAP. More specifically, we show that the closedness under maximisation of a family of stochastic exponentials claimed by \cite{WH}, as well as the approach itself of \cite{WH}, can only yield a proof of the FTAP in the context of complete financial markets, as already considered in \cite{LS}. Moreover, we provide an explicit counterexample showing that in general it is not possible to adapt the ideas of \cite{WH} in order to develop an alternative proof of the FTAP.

The paper is structured as follows. For the convenience of the reader, Section \ref{S2} recalls the financial market model considered in \cite{WH}. Section \ref{S3} critically analyses the proof of their main result, while Section \ref{S4} contains the counterexample. We refer to \cite{F,H} for a unified analysis of several no-arbitrage conditions in the context of  general continuous financial models and, more specifically, to the survey paper \cite{FR} for the properties of It\^o-process models when the NFLVR condition does not necessarily hold.

\section{The financial market model}	\label{S2}

On a given probability space $(\Omega,\F,P)$, let the process $W=\{W(t);0\leq t \leq T\}$ be an $\R^d$-valued Brownian motion, with $T\in(0,\infty)$ denoting a fixed time horizon, and let $\FF=(\F_t)_{0\leq t \leq T}$ be the $P$-augmented natural filtration of $W$.

It is assumed that $k+1$ (with $k\leq d$) securities are available for trade, with prices represented by the $\R^{k+1}$-valued continuous semimartingale $X=\{X(t);0\leq t \leq T\}$. As usual, the $0$th security denotes a locally riskless savings account process:
\[
X_0(t) = \exp\left\{\int_0^t\!r(u)\,\dd u\right\},
\qquad
\text{for all }t\in[0,T],
\]
for a progressively measurable interest rate process $r=\{r(t);0\leq t \leq T\}$ with $\int_0^T\!|r(t)|\dd t<\infty$ $P$-a.s. The remaining $k$ securities are risky and their prices $X_1,\ldots,X_k$ are given by the solutions to the following SDEs, for $i=1,\ldots,k$:
\be	\label{stock}	\left\{ \ba
dX_i(t) &= X_i(t)\,\mu_i(t)\,\dd t+X_i(t)\sum_{j=1}^d\sigma_{ij}(t)\,\dd W_j(t),	\\
X_i(0) &= x_i>0,
\ea	\right. \ee
with $\mu=\{\mu(t);0\leq t \leq T\}$ an $\R^k$-valued progressively measurable process with $\int_0^T\!\!\|\mu(t)\|\dd t<\infty$ $P$-a.s. and $\sigma=\{\sigma(t);0\leq t \leq T\}$ an $\R^{k\times d}$-valued progressively measurable process satisfying $\sum_{i=1}^k\!\sum_{j=1}^d\!\int_0^T\!\!\sigma_{ij}^2(t)\dd t<\infty$ $P$-a.s. (note that the square is missing in \cite{WH}). These assumptions ensure the existence of a unique strong solution to the system of SDEs \eqref{stock}. 
Furthermore, as in \cite{WH}, we suppose that the matrix $\sigma(t)$ has ($P$-a.s.) full rank for every $t\in[0,T]$, meaning that the financial market does not contain redundant assets (see also \cite{FR}, Remark 4.2.2)\footnote{We point out that the assumption that $\sigma(t)$ has $P$-a.s. full rank for every $t\in[0,T]$ can be relaxed by only assuming that $\mu(t)-r(t)I_k\in\{\sigma(t)x:x\in\R^d\}$ $P$-a.s. for every $t\in[0,T]$ (which corresponds to exclude pathological arbitrage possibilities known as \emph{increasing profits}; see \cite{FR}, Proposition 4.3.4), with $I_k:=(1,\ldots,1)^{\!\top}\in\R^k$, and by replacing the matrix $\sigma(t)^{\!\top}(\sigma(t)\,\sigma(t)^{\!\top})^{-1}$ with the Moore-Penrose pseudoinverse of $\sigma(t)$ in \eqref{MPR}.}.

\begin{definition}	\label{adm}
A progressively measurable process $\pi=\{\pi(t);0\leq t \leq T\}$ taking values in $\R^k$ is said to be an \emph{admissible trading strategy} if
\begin{subequations}
\be	\label{adm-1}
\int_0^T\bigl|\pi(t)^{\!\top}\bigl(\mu(t)-r(t)I_k\bigr)\bigr|\,\dd t<\infty \quad \text{$P$-a.s.}
\ee
\be	\label{adm-2}
\int_0^T\bigl\|\pi(t)^{\!\top}\sigma(t)\bigr\|^2\,\dd t<\infty \quad \text{$P$-a.s.}
\ee
\end{subequations}
and if the process $\int_0^{\cdot}X_0(t)^{-1}\,\pi(t)^{\!\top}\!\left(\mu(t)-r(t)I_k\right)\dd t
+\int_0^{\cdot}X_0(t)^{-1}\,\pi(t)^{\!\top}\sigma(t)\,\dd W(t)$ is $P$-a.s. uniformly bounded from below, with $^{\top}$ denoting transposition and $I_k:=(1,\ldots,1)^{\!\top}\in\R^k$.
\end{definition}

We denote by $\A$ the family of all admissible trading strategies. For $\pi\in\A$, the term $\pi_i(t)$ represents the amount of wealth invested on asset $i$ at time $t$, for $i=1,\ldots,k$ and $t\in[0,T]$. As usual, we assume that trading is done in a self-financing way. This amounts to requiring that the $X_0$-discounted wealth process $\bar{V}^{x,\pi}=\{\bar{V}^{x,\pi}(t);0\leq t \leq T\}$ associated to a strategy $\pi\in\A$, starting from an initial endowment of $x\in\R_+$, satisfies the following SDE:
\be	\label{wealth}	\left\{	\ba
\dd \bar{V}^{x,\pi}(t) &= X_0(t)^{-1}\,\pi(t)^{\!\top}\bigl(\mu(t)-r(t)I_k\bigr)\,\dd t
+X_0(t)^{-1}\,\pi(t)^{\!\top}\sigma(t)\,\dd W(t),	\\
\bar{V}^{x,\pi}(0) &= x.
\ea	\right. \ee
In turn, for $(x,\pi)\!\in\!\R_+\!\times\!\A$, the undiscounted wealth process $V^{x,\pi}\!=\!\{V^{x,\pi}(t);0\leq t \leq T\}$, defined by $V^{x,\pi}(t)\!:=\!X_0(t)\,\bar{V}^{x,\pi}(t)$, satisfies $V^{x,\pi}(0)\!=\!x$ and
\be	\label{wealth-undisc}
\dd V^{x,\pi}(t) = r(t)V^{x,\pi}(t)\,\dd t+\pi(t)^{\!\top}\bigl(\mu(t)-r(t)I_k\bigr)\,\dd t
+\pi(t)^{\!\top}\sigma(t)\,\dd W(t).
\ee

\begin{remark}
We want to point out that \cite{WH} do not require the integrability condition \eqref{adm-1}. If condition \eqref{NA1} holds, then \eqref{adm-1} follows from \eqref{adm-2}, as can be verified by a simple application of the Cauchy-Schwarz inequality (see e.g. \cite{FR}, Lemma 4.3.21). However, condition \eqref{adm-1} is in general necessary in order to ensure that the $\dd t$-integrals appearing in \eqref{wealth} and \eqref{wealth-undisc} are well-defined.
\end{remark}

\begin{definition}	\label{AO}
A trading strategy $\pi\in\A$ yields an \emph{arbitrage opportunity} if $P\bigl(V^{0,\pi}(T)\geq 0\bigr)=1$ and $P\bigl(V^{0,\pi}(T)>0\bigr)>0$.
\end{definition}

\begin{remark}
In \cite{WH}, a trading strategy $\pi$ is called \emph{tame} if $P\left(\bar{V}^{x,\pi}(t)\geq-1;\forall t\in[0,T]\right)=1$. As can be easily checked, there exist arbitrage opportunities with tame strategies if and only if there exist arbitrage opportunities with admissible trading strategies (in the sense of Definition \ref{adm} and according also to the terminology of \cite{DS1,DS2,DS3}).
\end{remark}

\section{An analysis of the main result of \cite{WH}}	\label{S3}

In this section, we critically analyse the proof of the FTAP given in \cite{WH}. We first show that stochastic exponentials do not possess the closure property stated in Lemma 3.3 of \cite{WH}, except for the particular case of a complete financial market.

\subsection{Market price of risk and martingale deflators}	\label{S3.1}

Since the matrix $\sigma(t)$ is assumed to be of full rank for every $t\in[0,T]$, we can define the \emph{market price of risk} process $\theta=\{\theta(t);0\leq t \leq T\}$ as follows, for all $t\in[0,T]$:
\be	\label{MPR}
\theta(t) := \sigma(t)^{\!\top}\bigl(\sigma(t)\,\sigma(t)^{\!\top}\bigr)^{-1}\bigl(\mu(t)-r(t)I_k\bigr).
\ee
As in Sections 3-4 of \cite{WH}, we introduce the following standing assumption on $\theta$:
\be	\label{NA1}
\int_0^T\!\!\|\theta(t)\|^2\,\dd t<\infty \quad \text{ $P$-a.s.}
\ee
Let us briefly comment on the implications of the above integrability condition.

\begin{definition}	\label{defl}
A strictly positive local martingale $Z\!=\!\{Z(t);0\!\leq\! t\! \leq\! T\}$ with $Z(0)\!=\!1$ is said to be a \emph{martingale deflator} if the product $Z\,\bar{V}^{x,\pi}$ is a local martingale, for all $\pi\!\in\!\A$ and $x\!>\!0$. We denote by $\D$ the set of all martingale deflators.
\end{definition}

As long as \eqref{NA1} holds, we can always define the strictly positive continuous local martingale $Z_0=\{Z_0(t);0\leq t\leq T\}$ as the stochastic exponential $Z_0:=\E\!\left(-\int\!\theta\,\dd W\right)$. A standard application of the integration by parts formula (see e.g. \cite{FR}, Proposition 4.3.9), together with \eqref{wealth} and \eqref{MPR}, yields that $Z_0\in\D$. This shows that $\D\neq\emptyset$ if condition \eqref{NA1} holds. 

\begin{remark}	\label{rem-NA1}
We want to point out that the existence of a martingale deflator has a precise meaning in the context of arbitrage theory. Indeed, as shown in Theorem 4 of \cite{Kar}, $\D\neq\emptyset$ holds if and only if there are \emph{No Arbitrages of the First Kind}, in the sense of Definition 1 of \cite{Kar} (compare also with \cite{F}, Theorem 5.4).
\end{remark}

The next result clarifies the importance of the class of stochastic exponentials considered in \cite{WH}. We denote by $K(\sigma)$ the family of all $\R^d$-valued progressively measurable processes $\nu=\{\nu(t);0\leq t \leq T\}$ such that $\int_0^T\!\|\nu(t)\|^2\dd t<\infty$ $P$-a.s. and $\sigma(t)\nu(t)=0$ $P$-a.s. for a.e. $t\in[0,T]$.

\begin{proposition}	\label{str-defl}
A strictly positive local martingale $Z=\{Z(t);0\leq t \leq T\}$ with $Z(0)=1$ belongs to $\D$ if and only if there exists an $\R^d$-valued process $\nu=\{\nu(t);0\leq t \leq T\}\in K(\sigma)$ such that, for all $t\in[0,T]$:
\be	\label{density}
Z(t) = \E\!\left(-\int\!(\theta+\nu)\,\dd W\right)_{\!\!t}
= Z_0(t)\,\E\!\left(-\int\!\nu\,\dd W\right)_{\!\!t} =: Z_{\nu}(t).
\ee
Moreover, an element $Z\in\D$ is the density process of an ELMM if and only if it is a (uniformly integrable) martingale, i.e., if and only if $E[Z(T)]=1$.
\end{proposition}
\begin{proof}
Recalling that the filtration $\FF$ is generated by $W$, the claim follows from the martingale representation theorem together with the integration by parts formula, using \eqref{wealth} and \eqref{MPR} (compare with \cite{FR}, Lemma 4.3.15). The last assertion follows from Bayes' rule together with the supermartingale property of the positive local martingale $Z_{\nu}$, for $\nu\in K(\sigma)$.
\end{proof}

In their study, \cite{WH} crucially rely on the following claim (see their Lemma 3.3): there exists a process $\lambda=\{\lambda(t);0\leq t\leq T\}\in K(\sigma)$ such that
\be	\label{wrongclaim}
Z_{\lambda}(t) = \underset{\nu\in K(\sigma)}{\esssup}\,Z_{\nu}(t),
\qquad \text{ for all }t\in[0,T].
\ee

Unfortunately, property \eqref{wrongclaim} fails to hold, except for the trivial case $K(\sigma)=\{0\}$, which, in view of Proposition \ref{str-defl}, is in turn equivalent to $\D=\{Z_0\}$. Indeed, in the proof of their Lemma 3.3, \cite{WH} attempt to show that the family $\{Z_{\nu}(t):\nu\in K(\sigma)\}$ is closed under pairwise maximisation, for all $t\in[0,T]$. However, this is not true, as we are now going to show.

Suppose that $K(\sigma)\neq\{0\}$ and, for two distinct elements $\nu_1,\nu_2\in K(\sigma)$, let us define the $\F_t$-measurable set $\chi(t):=\{Z_{\nu_1}(t)\geq Z_{\nu_2}(t)\}$, for every $t\in[0,T]$. Define then the process $\nu_3=\{\nu_3(t);0\leq t \leq T\}\in K(\sigma)$ as follows, for all $t\in[0,T]$:
\[
\nu_3(t) := \nu_1(t)\,\ind_{\chi(t)}+\nu_2(t)\,\ind_{\chi(t)^c}.
\]
It is clear that $Z_{\nu_3}(0)=1$ and, using the notation \eqref{density}:
\be	\label{dyn-1}
\dd Z_{\nu_3}(t) = -Z_{\nu_3}(t)\bigl(\theta(t)+\nu_3(t)\bigr)\,\dd W(t).
\ee
On the other hand, for all $t\in[0,T]$, we have:
\be	\label{dyn-2}
Z_{\nu_1}(t)\vee Z_{\nu_2}(t) = Z_{\nu_1}(t)\,\ind_{\chi(t)}+Z_{\nu_2}(t)\,\ind_{\chi(t)^c}
= Z_{\nu_1}(t)+\bigl(Z_{\nu_2}(t)-Z_{\nu_1}(t)\bigr)^+.
\ee
By the Tanaka-Meyer formula (see \cite{JYC}, Section 4.1.8), we get the following dynamics:
\be	\label{dyn-3}
\dd\bigl(Z_{\nu_2}(t)-Z_{\nu_1}(t)\bigr)^+
= \ind_{\{Z_{\nu_2}(t)>Z_{\nu_1}(t)\}}\,\dd\bigl(Z_{\nu_2}(t)-Z_{\nu_1}(t)\bigr)
+\frac{1}{2}\,L_{\nu_1,\nu_2}^0(t),
\ee
where $L_{\nu_1,\nu_2}^0=\{L_{\nu_1,\nu_2}^0(t);0\leq t \leq T\}$ denotes the local time of the continuous local martingale $Z_{\nu_2}-Z_{\nu_1}$ at the level $0$. From \eqref{dyn-2}-\eqref{dyn-3} we get:
\be	\label{dyn-4}	\ba
\dd\bigl(Z_{\nu_1}(t)\vee Z_{\nu_2}(t)\bigr)
&= -Z_{\nu_2}(t)\ind_{\chi(t)^c}\bigl(\theta(t)+\nu_2(t)\bigr)\,\dd W(t)	\\
&\quad\,
-Z_{\nu_1}(t)\ind_{\chi(t)}\bigl(\theta(t)+\nu_1(t)\bigr)\,\dd W(t)
+\frac{1}{2}\,L_{\nu_1,\nu_2}^0(t)	\\
&= -\bigl(Z_{\nu_1}(t)\vee Z_{\nu_2}(t)\bigr)\bigl(\theta(t)+\nu_3(t)\bigr)\,\dd W(t)
+\frac{1}{2}\,L_{\nu_1,\nu_2}^0(t).
\ea	\ee
By comparing \eqref{dyn-1} with \eqref{dyn-4}, we immediately see that $Z_{\nu_1}(t)\vee Z_{\nu_2}(t)\neq Z_{\nu_3}(t)$, thus contradicting the claim on page 193 of \cite{WH}. Note also that, as a consequence of Skorokhod's lemma (see \cite{JYC}, Lemma 4.1.7.1) it holds that $L_{\nu_1,\nu_2}^0(t)=\sup_{s\leq t}\bigl(-N(s)\bigr)$, for all $t\in[0,T]$, where the local martingale $N=\{N(t);0\leq t \leq T\}$ is defined by:
\[
N(t):=\int_0^t\sign\bigl(Z_{\nu_2}(u)-Z_{\nu_1}(u)\bigr)\,\dd\bigl(Z_{\nu_2}(u)-Z_{\nu_1}(u)\bigr),
\qquad\text{ for all }t\in[0,T].
\]
This implies that the local time $L_{\nu_1,\nu_2}^0$ appearing in \eqref{dyn-4} vanishes if and only if $Z_{\nu_1}=Z_{\nu_2}$, i.e., if and only if $\nu_1(t)=\nu_2(t)$ $P$-a.s. for a.e. $t\in[0,T]$. 
Indeed, if the latter condition holds, it is evident that $L_{\nu_1,\nu_2}^0$ vanishes. Conversely, if $L_{\nu_1,\nu_2}^0(t)=0$ $P$-a.s. for all $t\in[0,T]$, then the identity $L_{\nu_1,\nu_2}^0(t)=\sup_{s\leq t}\bigl(-N(s)\bigr)$ implies that the local martingale $N$ is $P$-a.s. non-negative and, hence, as a consequence of Fatou's lemma, it is a supermartingale. Since $N(0)=0$, this implies that $N=0$ $P$-a.s. In turn, $\langle Z_{\nu_2}-Z_{\nu_1}\rangle=\langle N\rangle=0$ $P$-a.s. implies that $\nu_1(t)=\nu_2(t)$ $P$-a.s. for a.e. $t\in[0,T]$. 
Together with \eqref{dyn-1} and \eqref{dyn-4}, we have thus shown that the trivial case $K(\sigma)=\{0\}$ is the only case where \eqref{wrongclaim} can hold.

More generally, the non-existence of a process $\lambda\in K(\sigma)$ satisfying \eqref{wrongclaim} can also be deduced from the following (almost trivial) result.

\begin{lemma}	\label{martingales}
Let $M_i=\{M_i(t);0\leq t \leq T\}$ be a right-continuous local martingale, for $i=1,2$, with $M_1(0)=M_2(0)$ $P$-a.s. Then $M:=M_1\vee M_2$ is a local martingale if and only if $M_1$ and $M_2$ are indistinguishable.
\end{lemma}
\begin{proof}
If $M=M_1\vee M_2$ is a local martingale, then the process $N=\{N(t);0\leq t \leq T\}$ defined by $N(t):=M(t)-\bigl(M_1(t)+M_2(t)\bigr)/2$, for all $t\in[0,T]$, is a non-negative local martingale and, by Fatou's lemma, also a supermartingale. Since $N(0)=0$ $P$-a.s., the supermartingale property implies $N(t)=0$ $P$-a.s. for all $t\in[0,T]$, meaning that $M_1(t)=M_2(t)$ $P$-a.s. for all $t\in[0,T]$. Right-continuity then implies that $P\bigl(M_1(t)=M_2(t),\forall t\in[0,T]\bigr)=1$.
The converse implication is evident.
\end{proof}

\begin{remark}
Besides the failure of the closedness under pairwise maximisation of the family $\{Z_{\nu}(t):\nu\in K(\sigma)\}$, the proof of Lemma 3.3 of \cite{WH} suffers from an additional technical problem. Indeed, in \cite{WH} the authors define a process $Z^*=\{Z^*(t);0\leq t \leq T\}$ as $Z^*(t):=\esssup_{\nu\in K(\sigma)}Z_{\nu}(t)$, for all $t\in[0,T]$, and a sequence of $[0,T]$-valued random variables $\{\tau_k\}_{k\in\N}$ by $\tau_k:=\inf\{t\in[0,T]:Z^*(t)\geq k\}\wedge T$, for $k\in\N$. However, a priori we do not know whether the process $Z^*$ is progressively measurable (or at least admits a right-continuous modification) and, hence, $\tau_k$ can fail to be a stopping time.
\end{remark}

\subsection{Attainable claims and the proof of Theorem 1.1 in \cite{WH}}	\label{S3.2}

The proof of the FTAP given in \cite{WH} relies on the characterisation of attainable claims. The notion of attainability implicitly adopted in Section 3 of \cite{WH} corresponds to the following definition. We always suppose that the standing assumption \eqref{NA1} holds true.

\begin{definition}
An $\F_T$-measurable non-negative random variable $B$ is said to be \emph{attainable} if there exist $\pi\in\A$ and $x\in\R_+$ such that $V^{x,\pi}(T)=B$ $P$-a.s. and if there exists a martingale deflator $Z_{\nu}\in\D$ such that $Z_{\nu}\bar{V}^{x,\pi}$ is a martingale.
\end{definition}

As long as $\D\neq\emptyset$, the following characterisation of attainable claims has been established in a general semimartingale setting in Theorem 3.2 of \cite{SY} (in the present setting, compare also with \cite{KLSX}, Theorem 8.5, and \cite{R}, Theorem 5 of Chapter 3). In particular, note that the existence of an ELMM is not necessarily assumed.

\begin{proposition}	\label{attain}
Let $B$ be an $\F_T$-measurable non-negative random variable. Then, under the standing assumption \eqref{NA1}, the following are equivalent:
\begin{enumerate}
\item[(i)]
$B$ is attainable;
\item[(ii)]
$\exists\,\lambda\in K(\sigma)$ such that 
\[
E\left[Z_{\lambda}(T)B/X_0(T)\right]
=\sup_{\nu\in K(\sigma)}E\left[Z_{\nu}(T)B/X_0(T)\right]<\infty.
\]
\end{enumerate}
Moreover, the initial endowment $x\in\R_+$ needed to replicate an attainable claim $B$ is given by $x=\sup_{\nu\in K(\sigma)}E\left[Z_{\nu}(T)B/X_0(T)\right]$.
\end{proposition}

In Section 4 of \cite{WH}, the implication \emph{(ii)}$\Rightarrow$\emph{(i)} of Proposition \ref{attain}, with $B=X_0(T)$, is used together with property \eqref{wrongclaim} in order to get the existence of an element $\pi\in\A$ such that $V^{c,\pi}(T)=X_0(T)$ $P$-a.s., starting from the initial investment $c:=E[Z_{\lambda}]$. If there does not exist any ELMM, then $c<1$ (see the last claim of Proposition \ref{str-defl}) and the portfolio process $V^{0,\pi}=\{V^{0,\pi}(t);0\leq t\leq T\}$ defined as $V^{0,\pi}(t):=V^{c,\pi}(t)-cX_0(t)$, for all $t\in[0,T]$, realises an arbitrage opportunity in the sense of Definition \ref{AO}, since $V^{0,\pi}(T)=(1-c)X_0(T)>0$ $P$-a.s.
This would prove that the absence of arbitrage opportunities (together with the standing assumption \eqref{NA1}, which excludes arbitrages of the first kind, see Remark \ref{rem-NA1}) implies the existence of an ELMM. Unfortunately, the fact that \eqref{wrongclaim} fails to hold, as shown in Section \ref{S3.1}, invalidates the proof of the FTAP proposed in \cite{WH}.

Actually, it can be shown that property \eqref{wrongclaim} and, hence, the arguments used in the proof of Theorem 1.1 in \cite{WH}, are valid if and only if the financial market is \emph{complete}, in the sense that every $\F_T$-measurable non-negative random variable $B$ such that $B/X_0(T)$ is bounded is attainable. 
Indeed, if the financial market is complete, the results of \cite{SY} (see also \cite{thes}, Theorem 4.5.13) imply that $\D=\{Z_0\}$. Due to Proposition \ref{str-defl}, the latter property is equivalent to $K(\sigma)=\{0\}$, in which case property \eqref{wrongclaim} trivially holds true. Conversely, if property \eqref{wrongclaim} holds, then the implication \emph{(ii)}$\Rightarrow$\emph{(i)} in Proposition \ref{attain} immediately shows that the financial market is complete.
In particular, if the financial market is complete, the payoff $X_0(T)$ can be replicated by a portfolio $V^{c,\pi}$ with $c=E[Z_0(T)]$ and $\pi\in\A$. In that case, the same arguments used in Section 4 of \cite{WH} allow to prove that the absence of arbitrage opportunities implies that $Z_0$ is the density process of the (unique) ELMM. We have thus shown that the proof of the FTAP proposed by \cite{WH} only works in the context of complete financial markets, as originally considered in \cite{LS}.

We want to remark that, in view of the implication \emph{(ii)}$\Rightarrow$\emph{(i)} of Proposition \ref{attain}, the proof of Theorem 1.1 in \cite{WH} does not actually need the full strength of the (untrue) property \eqref{wrongclaim}, but only the existence of a process $\lambda\in K(\sigma)$ such that $E[Z_{\lambda}(T)]=\sup_{\nu\in K(\sigma)}E[Z_{\nu}(T)]$. However, such a process does not necessarily exist in an incomplete market (if arbitrage opportunities are not a priori excluded), as will be shown in Section \ref{S4} by means of an explicit counterexample. This means that the approach of \cite{WH} (or an extension thereof) cannot yield an alternative proof of the FTAP in the context of general incomplete financial markets based on It\^o processes.

\section{A counterexample}	\label{S4}

This section exhibits a simple model of an incomplete financial market where the FTAP cannot be proved by relying on the techniques adopted in \cite{WH} (or an extension thereof). In view of the discussion at the end of the preceding section, we shall construct a market model for which there does not exist an element $\lambda\in K(\sigma)$ such that $E[Z_{\lambda}(T)]=\sup_{\nu\in K(\sigma)}E[Z_{\nu}(T)]$. More precisely, we will show that $E[Z_{\nu}(T)]<1$ for all $\nu\in K(\sigma)$, while $\sup_{\nu\in K(\sigma)}E[Z_{\nu}(T)]=1$.

Let $(\Omega,\F,\FF,P)$ be a given filtered probability space, where $\FF$ denotes the natural $P$-augmented filtration of a real-valued Brownian motion $W$ and $\F=\F_T$. For a fixed constant $a>0$, define the stopping time $\tau:=\inf\{t\in[0,T]:W(t)\geq a\}\wedge T$. It is clear that $P(\tau<T)>0$. Suppose that $r\equiv 0$ and let the price process $S$ of a single risky asset be given as the solution to the following SDE:
\be	\label{SDE-S}	\left\{	\ba
\dd S(t) &= \ind_{\{t>\tau\}}\bigl(1/S(t)\bigr)\,\dd t + \ind_{\{t>\tau\}}\dd W(t),	\\
S(0) &= 1.
\ea	\right.	\ee

\begin{lemma}	\label{prop-SDE}
The SDE \eqref{SDE-S} admits a unique strong solution $S=\{S(t);0\leq t \leq T\}$ that is $P$-a.s. strictly positive. Moreover, $E\left[1/S(T)|\F_{\tau}\right]<1$ $P$-a.s. on $\{\tau<T\}$.
\end{lemma}
\begin{proof}
Let us define the filtration $\FFtilde=(\Ftilde_t)_{0\leq t \leq T}$ by $\Ftilde_t:=\F_{(\tau+t)\wedge T}$, for $t\in[0,T]$, and the process $\Wtilde=\{\Wtilde(t);0\leq t \leq T\}$ by $\Wtilde(t):=W\bigl((\tau+t)\wedge T\bigr)-W(\tau)$, for $t\in[0,T]$. 
Clearly, the process $\Wtilde$ is continuous, $\FFtilde$-adapted and satisfies $\Wtilde(0)=0$.
Furthermore, the strong Markov property of the Brownian motion $W$ implies that $\Wtilde$ is also a Brownian motion (stopped at $T-\tau$) with respect to the filtration $\FFtilde$ (in particular, it is independent from $\Ftilde_0=\F_{\tau}$).
On the filtered probability space $(\Omega,\F,\FFtilde,P)$, let us consider the following SDE:
\be	\label{SDE-Bessel}	\left\{	\ba
\dd\Stilde(t) &= \bigl(1/\Stilde(t)\bigr)\,\dd t+\dd\Wtilde(t),	\\
\Stilde(0) &= 1.
\ea	\right.	\ee
The SDE \eqref{SDE-Bessel} admits a unique strong solution $\Stilde=\{\Stilde(t);0\leq t \leq T-\tau\}$, known as the three-dimensional Bessel process (see \cite{JYC}, Section 6.1), which satisfies $\Stilde(t)>0$ $P$-a.s. for all $t\in[0,T-\tau]$. Define then a process $S=\{S(t);0\leq t\leq T\}$ by $S(t):=\ind_{\{t\leq\tau\}}+\ind_{\{t>\tau\}}\Stilde(t-\tau)$, for all $t\in[0,T]$. Clearly, since $\tau$ is an $\FF$-stopping time and $\FF$ is right-continuous, the process $S$ is $\FF$-adapted, $P$-a.s. strictly positive and satisfies $S(0)=1$. Moreover, since $\Stilde$ is continuous:
\[	\ba
\dd S(t) &= \ind_{\{t>\tau\}}\dd\Stilde(t-\tau)	\\
&= \ind_{\{t>\tau\}}\bigl(1/\Stilde(t-\tau)\bigr)\,\dd t + \ind_{\{t>\tau\}}\dd\Wtilde(t-\tau)	\\
&= \ind_{\{t>\tau\}}\bigl(1/S(t)\bigr)\,\dd t + \ind_{\{t>\tau\}}\dd W(t),  
\ea	\]
thus showing that the process $S$ solves the SDE \eqref{SDE-S}. The second assertion of the lemma can be proved as follows, on the set $\{\tau<T\}$:
\be	\label{expect}
E\left[1/S(T)|\F_{\tau}\right]
= E\bigl[1/\Stilde(T-\tau)|\Ftilde_0\bigr]
= E\bigl[1/\Stilde(u)\bigr]\Bigr|_{u=T-\tau}
= 2\,\Phi\left(\frac{1}{\sqrt{T-\tau}}\right)-1<1,
\ee
where the second equality follows due to the independence of $\Wtilde$ and $\Ftilde_0$, together with the $\Ftilde_0$-measurability of $\tau$, and the third equality from Exercise 6.1.5.5 of \cite{JYC}, with $\Phi(\cdot)$ denoting the distribution function of a standard Gaussian random variable.
\end{proof}

According to the notation introduced in Section \ref{S3.1}, the market price of risk process $\theta$ is given by $\theta(t)=1/S(t)\ind_{\{t>\tau\}}$, for $t\in[0,T]$. Due to the continuity of $S$, the standing assumption \eqref{NA1} is satisfied, and hence, by It\^o's formula:
\be	\label{SDE-Z_0}
Z_0(T) = \exp\left(-\int_{\tau}^T\!\!\!\frac{1}{S(t)}\,\dd W(t)
-\frac{1}{2}\int_{\tau}^T\!\!\!\frac{1}{S(t)^2}\,\dd t\right)
= \frac{1}{S(T)}.
\ee
In the present context, a real-valued progressively measurable process $\nu=\{\nu(t);0\leq t \leq T\}$ belongs to $K(\sigma)$ if and only if it satisfies $\nu(t)=\nu(t)\ind_{\{t\leq\tau\}}$ $P$-a.s. for all $t\in[0,T]$ and $\int_0^{\tau}\!\nu(t)^2\dd t<\infty$ $P$-a.s. Proposition \ref{str-defl}, equation \eqref{SDE-Z_0} and Lemma \ref{prop-SDE} then imply the following, for every $\nu\in K(\sigma)$:
\[	\ba
E[Z_{\nu}(T)]
&= E\left[Z_0(T)\,\E\!\left(-\int\!\nu\,\dd W\right)_{\!\!\tau}\right]
= E\left[\frac{1}{S(T)}\,\E\!\left(-\int\!\nu\,\dd W\right)_{\!\!\tau}\right]	\\
&= E\left[E\left[\frac{1}{S(T)}\Bigr|\F_{\tau}\right]\E\!\left(-\int\!\nu\,\dd W\right)_{\!\!\tau}\right]
< E\left[\E\!\left(-\int\!\nu\,\dd W\right)_{\!\!\tau}\right]\leq 1,
\ea	\]
where the last inequality is due to the supermartingale property of $\E\left(-\int\!\nu\,\dd W\right)$, for $\nu\in K(\sigma)$. Due to the last assertion of Proposition \ref{str-defl}, this implies that the financial market model considered in this section does not admit an ELMM\footnote{Models based on three-dimensional Bessel processes are classical examples of financial markets that allow for arbitrage opportunities, see for instance \cite{DS2}, example 6.8 in \cite{thes} and the example on page 59 of \cite{FR}.}.

Moreover, there does not exist an element $\lambda\in K(\sigma)$ such that $E[Z_{\lambda}(T)]=\sup_{\nu\in K(\sigma)}E[Z_{\nu}(T)]$, as we are now going to show. For every $n\in\N$, let $\nu_n:=n\ind_{\lsi0,\tau\rsi}$. Clearly, we have $\nu_n\in K(\sigma)$ and $\E\left(-\int\!\nu_n\,\dd W\right)$ is a uniformly integrable martingale (stopped at $\tau$), for all $n\in\N$. Hence, due to Proposition \ref{str-defl} and equations \eqref{expect}-\eqref{SDE-Z_0}, we get, for any $n\in\N$:
\[	\ba
E[Z_{\nu_n}(T)]
&= E\left[Z_0(T)\,\E\!\left(-\int\!\nu_n\,\dd W\right)_{\!\!\tau}\right]
= E\left[\frac{1}{S(T)}\,\E\!\left(-\int\!\nu_n\,\dd W\right)_{\!\!\tau}\right]	\\
&= E\left[\E\!\left(-\int\!\nu_n\,\dd W\right)_{\!\!\tau}\right]
+E\left[\E\!\left(-\int\!\nu_n\,\dd W\right)_{\!\!\tau}\left(\frac{1}{S(T)}-1\right)\ind_{\{\tau<T\}}\right]	\\
&= 1-2\,E\left[\E\!\left(-\int\!\nu_n\,\dd W\right)_{\!\!\tau}\left(1-\Phi\Bigl(\frac{1}{\sqrt{T-\tau}}\Bigr)\right)\ind_{\{\tau<T\}}\right]	\\
&= 1-2\,E\left[\exp\biggl(-n\tau\Bigl(\frac{a}{\tau}+\frac{n}{2}\Bigr)\biggr)
\left(1-\Phi\Bigl(\frac{1}{\sqrt{T-\tau}}\Bigr)\right)\ind_{\{\tau<T\}}\right].
\ea	\]
By dominated convergence, the last equality implies that:
\[
1 = \lim_{n\rightarrow+\infty}E[Z_{\nu_n}(T)] \leq \sup_{\nu\in K(\sigma)}E[Z_{\nu}(T)] \leq 1,
\]
thus showing that $\sup_{\nu\in K(\sigma)}E[Z_{\nu}(T)]=1$. In particular, the supremum is not attained by any element $\nu\in K(\sigma)$.

\subsubsection*{Acknowledgements}
\begin{spacing}{1}
This research was supported by a Marie Curie Intra European Fellowship within the 7th European Community Framework Programme under grant agreement PIEF-GA-2012-332345. The author is thankful to two anonymous referees for valuable comments that helped to improve the paper. Part of this research was carried out while the author was a post-doctoral researcher at INRIA Paris-Rocquencourt, within the MATHRISK team.
\end{spacing}

\end{document}